\documentclass[conference,10pt]{IEEEtran}
\IEEEoverridecommandlockouts
\usepackage{cite}
\usepackage{amsmath,amssymb,amsfonts}
\usepackage{algpseudocode}
\usepackage{graphicx}
\usepackage{algorithm}
\usepackage{textcomp}
\usepackage{xcolor}
\usepackage{amsthm}
\def\BibTeX{{\rm B\kern-.05em{\sc i\kern-.025em b}\kern-.08em
    T\kern-.1667em\lower.7ex\hbox{E}\kern-.125emX}}
    
\theoremstyle{definition}
\newtheorem{scenario}{Scenario}

\theoremstyle{definition}
\newtheorem{definition}{Definition}

\theoremstyle{remark}
\newtheorem{case}{Case}

\theoremstyle{remark}
\newtheorem{problem}{Problem}

\theoremstyle{definition}
\newtheorem{proposition}{Proposition}

\theoremstyle{remark}
\newtheorem{solution}{Solution}

\theoremstyle{definition}
\newtheorem{theorem}{Theorem}

\begin{document}

\title{Crash Consistency in DRAM-NVM-Disk Hybrid Storage System


}

\author{\IEEEauthorblockN{1\textsuperscript{st} Guoyu Wang}
\IEEEauthorblockA{\textit{College of Computer Science and Technology} \\
\textit{Jilin University}\\
Changchun, China \\
wgy21@mails.jlu.edu.cn
}
\and
\IEEEauthorblockN{2\textsuperscript{nd} Xilong Che}
\IEEEauthorblockA{\textit{College of Computer Science and Technology} \\
\textit{Jilin University}\\
Changchun, China \\
chexilong@jlu.edu.cn}
\and
\IEEEauthorblockN{3\textsuperscript{rd} Haoyang Wei}
\IEEEauthorblockA{\textit{College of Computer Science and Technology} \\
\textit{Jilin University}\\
Changchun, China \\
hywei23@mails.jlu.edu.cn}
\and
\IEEEauthorblockN{4\textsuperscript{th} Chenju Pei}
\IEEEauthorblockA{\textit{College of Computer Science and Technology} \\
\textit{Jilin University}\\
Changchun, China \\
peicj2121@mails.jlu.edu.cn}
\and
\IEEEauthorblockN{5\textsuperscript{th} Juncheng Hu$^*$}
\IEEEauthorblockA{\textit{College of Computer Science and Technology} \\
\textit{Jilin University}\\
Changchun, China \\
jchu@jlu.edu.cn} \thanks{$^*$Corresponding author. }
}

\maketitle

\begin{abstract}
NVM is used as a new hierarchy in the storage system, due to its intermediate speed and capacity between DRAM, and its byte granularity. However, consistency problems emerge when we attempt to put DRAM, NVM, and disk together as an efficient whole. In this paper, we discuss the challenging consistency problems faced by heterogeneous storage systems, and propose our solution to the problems. The discussion is based on NVPC as a case study, but can be inspiring and adaptive to all similar heterogeneous storage systems. 

\end{abstract}


\section{Introduction}
With the emergence of non-volatile memory (NVM), various works have been proposed to leverage its byte-addressable and persistence characteristics. 
Among these approaches, the most compatible, efficient, and promising one is to use NVM to accelerate existing disk file systems \cite{linP2CACHEExploringTiered2023, wooStackingPersistentMemory2023, wangNVPCTransparentNVM2024}, in which case users can utilize NVM in their current system transparently without migration cost. However, we discover that the heterogeneity between DRAM, NVM, and disk makes it harder for these works to fully exploit the potential of NVM. 

NVPC \cite{wangNVPCTransparentNVM2024} is the only one of the accelerators that precisely accelerates only the slow path, i.e. sync writes, of the legacy storage stack, with no slow-down to its lower file system. We notice that its implementation abandons the layered design adopted by previous work, introducing an on-demand sync absorbing mechanism, which is the fundamental of its high efficiency. However, this design also introduces consistency problems between NVM and disk, which are only briefly and insufficiently discussed in that work. We observe that this consistency problem is common to a heterogeneous storage system with different write sequences and access granularity. Thus we believe that it is necessary to further discuss the consistency problems in such systems. 

In this work, we will first introduce the background of NVM-based heterogeneous storage systems.
Then based on NVPC, we will discuss the consistency problems and solutions. Specifically, for each heterogeneous characteristic, we will first figure out the exact problems that the system is facing. Then we will analyze the causes of the consistency problems. Finally, we will provide the solutions to each problem based on our analysis.

\section{Heterogeneous Storage Systems}

As the new hierarchy of the storage system, NVM has a byte-addressable access pattern, with intermediate speed and capacity between DRAM and SSD. The unique characteristics of NVM bring both chances and challenges to its users. Various works have been proposed to leverage NVM, such as NVM-based file systems \cite{dulloorSystemSoftwarePersistent2014, xuNOVALogstructuredFile2016, dongPerformanceProtectionZoFS2019}, tiered memories \cite{marufTPPTransparentPage2023, weinerTMOTransparentMemory2022}, and NVM-enabled databases \cite{jiFalconFastOLTP2023, ruanPersistentMemoryDisaggregation2023, wangPacmanEfficientCompaction2022}. However, when we put different storage layers with different characteristics together, the challenges become particularly noticeable. 

The early attempts to integrate heterogeneous storage devices are cross-media file systems. For example, Strata \cite{kwonStrataCrossMedia2017} and Ziggurat \cite{zhengZigguratTieredFile2019} are designed as monolithic file systems that use both DRAM, NVM, and disk, to provide a fast and a large capacity in the same time. Strategies are provided in these works to place data on the most suitable storage device. However, these cross-media file systems are usually complex and less mature than traditional disk file systems. The potential security risk and the data migration cost hinder their wide use. 

A novel approach is to utilize NVM as a transparent disk file system accelerator. SPFS \cite{wooStackingPersistentMemory2023} is an NVM-based overlay file system atop the disk file system, providing sync write prediction and absorption to eliminate the slow foreground I/O hang. Since SPFS is a separate layer before the DRAM page cache, when a write is executed, it should either be absorbed by SPFS or be left to the lower page cache and disk file system, at which point we don't know if there will be a sync next. If the decision is wrong, we may suffer from the long disk I/O on a false-positive, or experience unnecessary slow NVM write on a false-negative. The two-tier indexing also introduces extra slow-down to disk file systems. 

P2CACHE \cite{linP2CACHEExploringTiered2023} is another NVM-based overlay accelerator, trying to adopt both the fast read speed of DRAM and the persistency of NVM. The difference from SPFS is that P2CACHE puts DRAM and NVM into one layer. Each write on P2CACHE is stored to both the DRAM and the NVM. This double-write mechanism is easy to implement, and can provide a strong consistency. However, the cost of accelerating sync writes is the slow-down to all normal writes. 

NVPC \cite{wangNVPCTransparentNVM2024} is the latest work that provide transparent disk file system acceleration with no extra performance tax. The key of its efficient acceleration is the balance between heterogeneous storage devices. Specifically, reads, writes, and disk write-backs are always served by the fast DRAM. Only the slow sync will log relevant writes to the NVM. Meanwhile, NVPC provide an active sync strategy to efficiently absorb small sync writes with the byte-addressing ability of NVM. 

Though NVPC is the only accelerator that introduces no slow-down, the performance doesn't come for free. The NVM-absorbed sync writes, the async write-back to disk, the different write granularity between disk and NVM, and the different access modes (queued v.s. direct) between disk and NVM, lead to great challenges of the eventual consistency of data. The consistency problems and the solutions are only briefly and partially addressed in the paper of NVPC. We find that the problems are more complex than they seem to be, and worth to be discussed further. We are taking NVPC as a case study in this paper, but we believe that the discussions are also meaningful to the design of all heterogeneous storage systems that adopt different devices with different write sequences and different write granularity.

\section{Consistency Problems and Solutions in the System}

To analyze the consistency problems faced by heterogeneous storage systems, we need to first clarify the definition and the scope of a sync operation. Sync (e.g. fsync) on Linux or other POSIX-compatible systems forces all data for a file to be transferred to the storage device. The sync will not return until the data is successfully applied to the storage, at which point the storage reaches a consistent state with the DRAM page cache. For the special case, O\_SYNC, the difference is that there is a sync operation implicitly following each write operation on a file descriptor. However, its consistency model shows no difference with an explicit sync. We do not discuss the multi-descriptor or multi-thread behavior of sync and normal writes here, because it is somewhat undefined, and such ambiguity should be eliminated with locks by users. 

From the user's aspect, sync is usually used as a "barrier" on the storage system, implying that \textit{any event happens after the return of the sync can be based on the premise that the data before the sync has been persisted by the storage}. We will focus on the file system events, i.e. read, write, sync, and write-back here, because any other events that may influence the external visibility are based on the promise of the file system.

We also need to define what is a crash for the consistency discussion. When we refer to a crash, it usually means that the power is down at some point of time. After that, all volatile states of the system are vanished, and the only way to grab data is to access persistent storage devices. In the discussion scope of NVPC, the significant feature of crash is that after the crash, all page cache are lost, and NVPC needs to recover data from NVM and disk. 


\begin{definition}\label{def:0}
The event type set $E=\{read(r),\allowbreak write(w),\allowbreak sync,\allowbreak writeback(wb),\allowbreak crash,\allowbreak "(,)",\allowbreak "[,]"\}$.
\end{definition}

\begin{definition}\label{def:1}
The sequence "$(,)$" on the set of events of a file system defines such a relation: $a,b,c\in E$, if $a$ happens before $b$, and $b$ happens before $c$, and there are no $crash$ event besides $a$, $b$, and $c$, then $(a,b,c)$. The number of events in "$(,)$" is variable. 
Note that "$(,)$" only defines the order of the events inside, but does not imply that the events should happen adjacently, i.e, $(read, write, sync, crash, write) \Leftrightarrow (read, crash, write)$.
\end{definition}

\begin{definition}\label{def:2}
The sequence "$[,]$" on the set of events of a file system defines such a relation: if $(a,b,c)$, and any event $X$ besides $a$, $b$, and $c$ between the first element $a$ and the last element $c$ satisfy $X\not\in \{write,sync,writeback\}$, then $[a,b,c]$. The number of events in "$[,]$" is variable. 
\end{definition}

\begin{theorem}\label{theo:1}
If $(write,sync,crash,read)$, then when $read$ is performed, the data on the storage should be no earlier than the $write$.
\end{theorem}

\begin{theorem}\label{theo:2}
The sync semantics of a file system are not violated, if and only if for each $(write,sync,crash,read)$, Theorem \ref{theo:1} is respected.
\end{theorem}

Theorem \ref{theo:1} and Theorem \ref{theo:2} are axiomatic according to the definition of sync operation. 

In NVPC or other similar heterogeneous storage systems, we would expect that the sync semantics are not violated. However, since we have two (or more) write destination devices with different write sequences and different access granularity, consistency problems emerge. The following discussion will be based on NVPC, but the same problems and solutions also apply to other heterogeneous storage systems.

\subsection{Problems from Write Sequence Difference}\label{sec:seq}

\setcounter{scenario}{0}
\setcounter{case}{0}

To maximize the performance of the system while maintaining the capacity, some data should be written to the NVM, and some data to the disk. Specifically, in NVPC, those write requests with urgent persist needs are directed to the NVM, while other writes obey the async write-back rule from the DRAM page cache to the disk. This leads to different write sequences on NVM and disk. In this section, we suppose that the NVM and the disk have the same write granularity, which is a page, to demonstrate the problems brought by different write sequences.

\begin{figure*}[t]
\centerline{\includegraphics{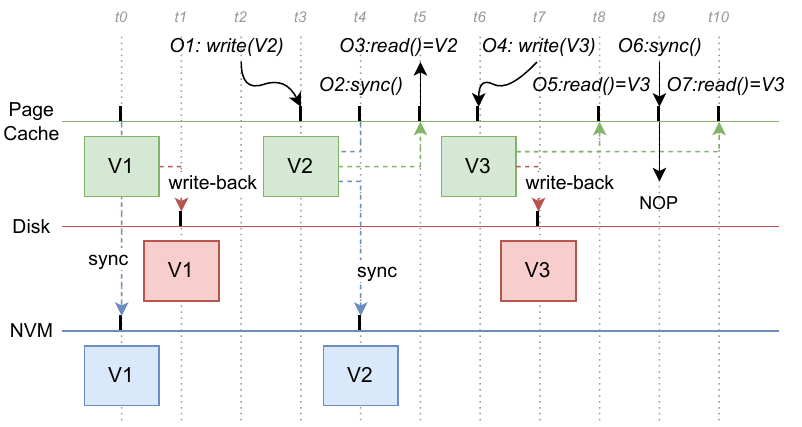}}
\caption{Space-time diagram for problematic cases caused by variant write sequence. }
\label{fig1}
\end{figure*}

Fig. \ref{fig1} shows a workflow on NVPC. At $t2$, both page cache, disk, and NVM reach a consistent state. Then as we perform operations on NVPC, data flows to different places as expected. However, when we inject crashes into the timeline, we will find some problematic cases caused by the different write sequences between disk and NVM. 

\begin{scenario}\label{scene11}
If we crash at $t5$, then after recovery, we have $V1$ on disk and $V2$ on NVM. It seems like NVM has successfully stored the newer data version ($V2$) because of the sync ($O2$). However, since we don't have a timestamp on the data, we don't know whether the disk version ($V1$) or the NVM version ($V2$) is fresher. From the diagram we can see that $V2$ is fresher, so let's assume that we adopt the data on NVM after a crash. In such case, if we perform $O3$ after the crash, we have $(O1,sync,crash,O3)$, and $O3$ returns $V2$ from the NVM, which respects Theorem \ref{theo:1}. Thus according to Theorem \ref{theo:2}, the sync semantics are not violated. 
\end{scenario}

\begin{scenario}\label{scene12}
If we crash at $t8$, then after recovery, we have $V3$ on disk and $V2$ on NVM. Now the disk has the fresher version. If we follow the assumption in Scenario \ref{scene11}, then after the crash when we perform $O5$ we adopt $V2$ on NVM instead of $V3$ on disk. Now we encounter a data rollback because we are replacing newer $V3$ with older $V2$. But it still doesn't violate the sync semantics, because there is no sync between $O4$ and $O5$, i.e. $(O4,\lnot sync,crash,O5)$.
\end{scenario}

\begin{scenario}\label{scene13}
If we crash at $t10$, then after recovery, we have $V3$ on disk and $V2$ on NVM, just like Scenario \ref{scene13}. However, the difference is that we performed a sync ($O6$) before the crash. Since the page is not dirty at $t9$, $O6$ will return with nothing done. But when we perform $O7$ after the crash, following the assumption in Scenario \ref{scene11}, we will get $V2$ from the NVM instead of the $V3$ on disk. This violates Theorem \ref{theo:1}, because $(O4,sync,crash,O7)$, but the data of $O4$ ($V3$) is lost if we perform $O7$ after a crash. The only way to fix this is to adopt the data version on the disk, but this in turn leads to the violation of Theorem \ref{theo:1} in Scenario \ref{scene11}.
\end{scenario}

\begin{problem}\label{prob1}
The different write sequences of disk and NVM, e.g. $V1\rightarrow V3$ and $V1\rightarrow V2$ in Fig. \ref{fig1}, may lead to two different final versions of data on the two storage devices. Since we don't know which one is fresher, choosing any one of them as the final version can violate Theorem \ref{theo:1}, and thus violates the sync semantics. 
\end{problem}

\begin{solution}\label{sol1}
Since Problem \ref{prob1} is caused by the missing version information for disk and NVM data, we can solve it by adding the lost information. Specifically, we can maintain a global persistent variable in the NVM for each page to indicate which storage device has the latest data version of that page. Then we can choose the right data according to this information after the crash. The process is described in Algorithm \ref{alg1}.
\end{solution}

\algdef{SE}[VARIABLES]{Variables}{EndVariables}
   {\algorithmicvariables}
   {\algorithmicend\ \algorithmicvariables}
\algnewcommand{\algorithmicvariables}{\textbf{global variables}}
\algnewcommand{\LineComment}[1]{\State \(\triangleright\) #1}

\begin{algorithm}
\caption{Solution to problem 1: maintain a persistent variable for each data page. }
\label{alg1}

\begin{algorithmic}[1]

\Variables
    \State $latest\_dev[1\ldots page\_num]$ \Comment{Device that stores the latest data. Persistent.}
\EndVariables

\Procedure{writeback}{$page$}

    \State \Call{write\_disk}{$page$}
    \State $latest\_dev[page.index] \gets \mathrm{DISK}$ 

\EndProcedure

\Procedure{\textsc{sync}}{$page$}

    \If{$page$ is dirty}
        \State \Call{write\_nvm}{$page$}
        \State $latest\_dev[page.index] \gets \mathrm{NVM}$ 
    \EndIf
\EndProcedure

\Function{crash\_recover}{$page\_index$}
    \If{$latest\_dev[page\_index] = \mathrm{DISK}$ }
        \State \Return \Call{read\_disk}{$page\_index$}
    \ElsIf{$latest\_dev[page\_index] = \mathrm{NVM}$ }
        \State \Return \Call{read\_nvm}{$page\_index$}
    \EndIf
\EndFunction

\end{algorithmic}
\end{algorithm}

\begin{proposition}\label{prop1}
Solution \ref{sol1} solves Problem \ref{prob1}. 
\end{proposition}

\begin{figure*}[t]
\centerline{\includegraphics{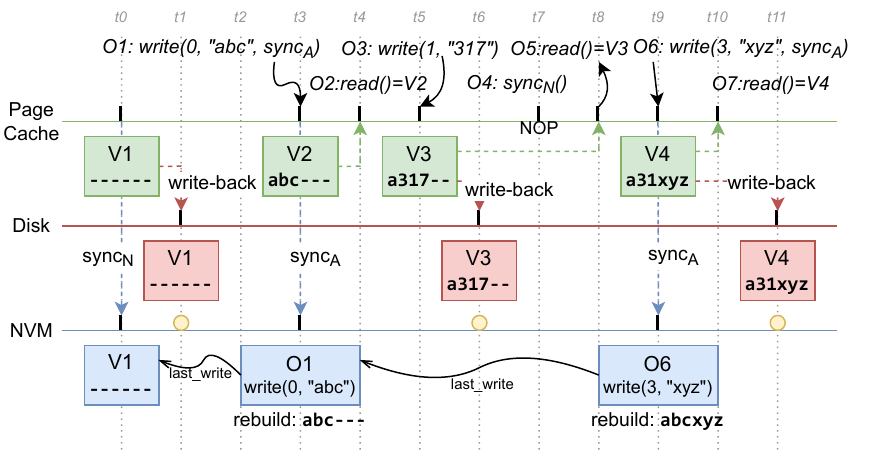}}
\caption{Space-time diagram for problematic cases caused by variant write granularity. } 
\label{fig2}
\end{figure*}

\begin{proof}
To prove that all $(w,sync,crash,r)$ respect Theorem \ref{theo:1}, we can simply prove that the last $(w,sync,crash,r)$ respects Theorem \ref{theo:1}, because when $w$ is persisted, the whole page with the data version no earlier than $w$ is persisted, meaning that the results of all previous writes are persisted. 
Now $(w,sync,crash,r)$ has and only has the following three possible cases to prove in the context of Problem 
\ref{prob1}:

\begin{case}\label{case11}
A write-back $wb$ happens after the write $w$ before the $sync$ operation is called. We mark it as $[w,wb,sync,crash,r]$. In this case, according to Algorithm \ref{alg1}, $wb$ writes the data of $w$ to the disk, marks the page as clean, and marks the page's $latest\_dev$ as DISK. The $sync$ needs to do nothing on the already-cleaned page. Then after the $crash$, the $r$ will read the recovered data from DISK, which is the data of $w$, respecting Theorem \ref{theo:1}.
\end{case}
\begin{case}\label{case12}
A write-back $wb$ happens after the $sync$ before the $crash$. We mark it as $[w,sync,wb,crash,r]$. In this case, the $sync$ first writes the data of $w$ to the NVM and marks the page's $latest\_dev$ as NVM. Then $wb$ writes the data again to the disk and marks the page's $latest\_dev$ as DISK. Then after the $crash$, the $r$ will read the recovered data from DISK, which is the data of $w$, respecting Theorem \ref{theo:1}.
\end{case}
\begin{case}\label{case13}
The $crash$ occurred before the write-back $wb$ would have a chance to perform. We mark it as $[w,sync,crash,r]$. In this case, the $sync$ writes the data of $w$ to the NVM and marks the page's $latest\_dev$ as NVM. Then after the $crash$, the $r$ will read the recovered data from NVM, which is the data of $w$, respecting Theorem \ref{theo:1}.
\end{case}

According to Theorem \ref{theo:2}, now that all cases respect Theorem \ref{theo:1}, Algorithm \ref{alg1} is proved to solve Problem \ref{prob1}.
\end{proof}

\subsection{Problems from Write Granularity Difference}
\label{sec:grn}

\setcounter{scenario}{0}
\setcounter{case}{0}

To leverage the byte-addressable characteristic of NVM, NVPC introduces small sync writes. Thus on NVM, data can be written with any length, which is different from the block granularity of disk. Such arbitrary write length support can improve the performance of small sync writes, but brings further consistency problems to the system.

Fig. \ref{fig2} shows a workflow on NVPC with arbitrary write length support. At $t2$, both page cache, disk, and NVM reach a consistent state. Then as we perform operations on NVPC, data flows to different places correctly. The difference with Fig. \ref{fig1} is that the write to the NVM is no longer fixed to a page. For each NVM record, it can persist a sync write event with any write length below a page. Note that for ease of management, the data should not cross page boundaries, or it should be divided into multiple records according to page boundaries. Also note that NVPC will perform a whole-page sync before starting a series of arbitrary length sync writes, so that previous writes are not lost.
We can mark whole-page normal sync as $sync_N$, arbitrary-length sync write as $[w,sync_A]$, and a series of arbitrary length sync writes as $(sync_N,\lnot w,[w,sync_A]^n),n\ge 1$.
To rebuild the whole page data from NVM, NVPC has to backtrack all partial writes on the same page. Now when we inject crashes into the timeline, more problematic cases occur, due to the different write granularity between disk and NVM.

\begin{scenario}\label{scene21}
Without Solution \ref{sol1}, the workflow in Fig. \ref{fig2} still faces the same consistency problems as Scenario \ref{scene11} and Scenario \ref{scene12}. E.g. when a crash happens at $t4$ or $t8$, NVPC will have trouble with choosing data between the disk version and the NVM version. By adding Solution \ref{sol1}, these problems can be solved similarly. The proof is the same and is omitted. 
\end{scenario}

\begin{scenario}\label{scene22}
Suppose that we have adopted Solution \ref{sol1}. Now if we crash at $t10$, then after the recovery we have $V3$ (\texttt{a317-{}-}) on disk, and $V1\rightarrow O1\rightarrow O6$ on NVM. By Algorithm \ref{alg1} we will choose the NVM as the latest data version, and from the NVM we can rebuild \texttt{abcxyz}. We have $(write_{O6},sync_{O6},[crash,O7])$ and the data of $O6$ (\texttt{xyz}) is successfully applied to the final version data, so that $O7$ can read it correctly, respecting Theorem \ref{theo:1}. However, the final result \texttt{abcxyz} is actually violating Theorem \ref{theo:1}. 
Because the arbitrary write length on NVM means that $O6$'s write may not overwrite previous write $O3$, and we need to also check $(O3,sync_{O4},crash,O7)$ in this workflow.
When we rebuild \texttt{abcxyz} from the NVM we only have the outcome of $O1$ and $O6$, but $O3$'s data \texttt{317} is lost. Thus $(O3,sync_{O4},crash,O7)$ violates Theorem \ref{theo:1} even if Solution \ref{sol1} is adopted.
\end{scenario}

\begin{problem}\label{prob2}
The problem we are facing in Scenario \ref{scene22} is that the write sequence and the write granularity of NVM are both different from disk, e.g. in Fig. \ref{fig2}, the disk has $V1\Rightarrow V3$, and the NVM has $V1\Rightarrow O1$(part of $V2$)$\Rightarrow O6$(part of $V4$). Thus when we retrieve data from NVM, sync data that is persisted on disk by write-back may lost, because later NVM records do not take charge of previous writes that are out of its range. E.g. at $t6$ $O3$ is successfully recorded by the disk, while at $t9$ the NVM only records $O6$, but not the full $V4$, leading to the loss of $O3$.
\end{problem}

\begin{solution}\label{sol2}
It is worth noting that an NVM sync write does not record data out of its range, but a disk write-back records all current data of a page, including any previous writes that may influence the current data version. Thus we don't need to worry about the disk, and should focus on the unrecorded hollows of a page on the NVM. Now we record each disk write-back event on the NVM, like the yellow bubbles shown in Fig. \ref{fig2}. The bubble represents that at that point of time the disk has maintained a full data version of the page, including the results of all previous writes. So when we backtrack the NVM records, whenever we meet a write-back event we will know that there is a reliable data copy on the disk, and we don't need to backtrack on the NVM any further. Then to rebuild the latest data version, we just apply the tracked useful NVM operations to the disk data version. This process is described in Algorithm \ref{alg2}. 
\end{solution}

\begin{algorithm}
\caption{Solution to problem 2: record write-back events on the NVM. }
\label{alg2}

\begin{algorithmic}[1]

\Procedure{writeback}{$page$}

    \State \Call{write\_disk}{$page$}
    \State $record.type\gets \mathrm{WRITEBACK}$ 
    \State $record.page\_index\gets page.index$
    \State \Call{write\_nvm}{$record$}

\EndProcedure

\Procedure{sync\_write}{$page\_index$, $data$, $off$, $len$}
    
    \State $record.type\gets \mathrm{WRITE}$ 
    \State $record.page\_index\gets page\_index$
    \State $record.data\gets data$
    \State $record.off\gets off$
    \State $record.len\gets len$
    \State \Call{write\_nvm}{$record$}

\EndProcedure

\Function{crash\_recover}{$page\_index$}
    \State $record \gets $ \Call{last\_nvm\_record\_on}{$page\_index$}
    \State $useful\_records \gets [\;]$
    \State $i \gets 0$
    \While{$record \ne \mathrm{NULL}$ \textbf{and} $record.type \ne \mathrm{WRITEBACK}$}
        \State $useful\_records[i] \gets record$
        \State $i = i+1$
        \State $record \gets $ \Call{prev\_nvm\_record\_on\_same\_idx}{\allowbreak$record$}
    \EndWhile
    \State $page \gets $ \Call{read\_disk}{$page\_index$}
    \For{$i > 0$}
        \State $i = i-1$
        \State \Call{write\_page}{$page$, $record.data$, $record.off$, $record.len$}
    \EndFor
    \State \Return $page$
\EndFunction

\end{algorithmic}
\end{algorithm}

\begin{proposition}\label{prop2}
Solution \ref{sol2} solves both Problem \ref{prob1} and Problem \ref{prob2}. 
\end{proposition}

\begin{proof}
Solution \ref{sol2} is a superset of Solution \ref{sol1}. It is trivial to see that if we apply Solution \ref{sol2} to Problem \ref{prob1}, the cases in the proof of Proposition \ref{prop1} can still be proved in the same way. Thus for Problem \ref{prob1}, we can use Solution \ref{sol2} as a stronger substitute for Solution \ref{sol1} rather than a supplement.

Besides the cases in the proof of Proposition \ref{prop1}, we need to further examine the results for multiple arbitrary-length sync writes, because on the NVM these sync writes may not cover the whole page in the context of Problem \ref{prob2}. Note that write-back still handles all writes before it, thus only the latest write-back needs to be further considered. There are four special cases of $(w,sync,crash,r)$ left to be proved:

\begin{case}\label{case21}
$[w_0,sync_N,[w_1,sync_A]^n,crash,r], n\ge 1$. This means that write-back never happens during the workflow before the crash, and at least one pair of arbitrary-length sync writes happens before the crash. In this case, NVPC backtracks each sync record and rebuilds the whole-page data according to its design, thus each $(w,sync,crash,r)$ can be recovered, which respects Theorem \ref{theo:1}. 
\end{case}

\begin{case}\label{case22}
$[w_0,wb,sync_N,[w_1,sync_A]^n,crash,r], n\ge 1$.
This means that write-back happens before a series of arbitrary-length sync writes and their initial $sync_N$. In this case, the $(w_1,sync_A,crash,r)$ has been proved in Case \ref{case21}. So we should focus on $(w_0,sync,crash,r)$. According to Algorithm \ref{alg2}, $wb$ writes the full-page data to the disk, and marks a write-back event on the NVM, so that after the crash, the recovered arbitrary-length sync writes will be performed based on the $wb$ version. The $wb$ version data includes $w_0$, so $(w_0,sync,crash,r)$ respects Theorem \ref{theo:1}. Thus each $(w,sync,crash,r)$ in this case respects Theorem \ref{theo:1}. 
\end{case}

\begin{case}\label{case23}
$[w_0,sync_N,wb,[w_1,sync_A]^n,crash,r], n\ge 1$. This means that write-back happens just after the initial $sync_N$ of a series of arbitrary-length sync writes. In this case, $(w_0,sync,crash,r)$ respects Theorem \ref{theo:1} because, just like we proved in Case \ref{case22}, the $wb$ writes its previous writes including $w_0$, which can be recovered after crash. Other $(w_1,sync_A,crash,r)$ has been proved in Case \ref{case21}. Thus each $(w,sync,crash,r)$ in this case respects Theorem \ref{theo:1}. 
\end{case}

\begin{case}\label{case24}
$[w_0,sync_A,wb,[w_1,sync_A]^n,crash,r], n\ge 1$. This means that write-back happens between a series of arbitrary-length sync writes. In the same way as Case \ref{case22}, $wb$ makes $(w_0,sync,crash,r)$ respect Theorem \ref{theo:1}. In the same way as Case \ref{case21}, $(w_1,sync_A,crash,r)$ respects Theorem \ref{theo:1}. Thus each $(w,sync,crash,r)$ in this case respects Theorem \ref{theo:1}.
\end{case}

Proof to other cases, e.g. write-back happens after the last sync, will be the same with the proof of Proposition \ref{prop1} even under the context of Problem \ref{prob2}. According to Theorem \ref{theo:2}, now that all cases respect Theorem \ref{theo:1}, Algorithm \ref{alg2} is proved to solve Problem \ref{prob2}.
\end{proof}

\subsection{Problems from Write-back Duration}
\label{sec:dur}
\setcounter{scenario}{0}
\setcounter{case}{0}

We should be aware that the write-back process does not happen at a point of time, but is a duration of time. Due to the slow speed of the disk, the legacy storage stack adopts a queue-based design to decouple the write process from the host side and the device side. The queue mechanism can largely improve the I/O throughput, but it brings challenges to NVPC. Specifically, we can only know the start and the end time points of a write-back event, but we can't know the exact time point that the data is persisted to the device. 

A write-back process can be described as follows: First, the page is marked as write-back, and enters the I/O queue, which is the start of write-back. Then, the disk consumes the pages in the queue and persists them properly. After the page is successfully persisted, the disk signals a finish event to tell the host, which is the end of the write-back. The implementation varies, but a write-back process always obeys the following principles: 

\begin{enumerate}
    \item The real write-disk event ($wb^r$) happens after the start of write-back ($wb^s$) before the end of write-back ($wb^e$), and a page can only be in one write-back event. I.e. $(wb^s,\lnot wb,wb^r,\lnot wb,wb^e)$ is always correct. 
    \item The data written to the disk is the data version that the real write-disk event ($wb^r$) happens. Since we don't know the exact time of $wb^r$, we can only be sure that the data version of the $wb$ is not earlier than the last data before $wb^s$, and is not later than the first data after $wb^e$. 
    \item At $wb^s$ the \texttt{DIRTY} state of a page is cleared, and the \texttt{WRITEBACK} state is set, within a pair of lock / unlock. The \texttt{WRITEBACK} state is cleared at $wb^e$. During the write-back ($wb^s$ to $wb^e$), if the page is written again, the \texttt{DIRTY} state is set again, which will not be cleared until the next $wb^s$. 
\end{enumerate}

\begin{figure*}[t]
\centerline{\includegraphics{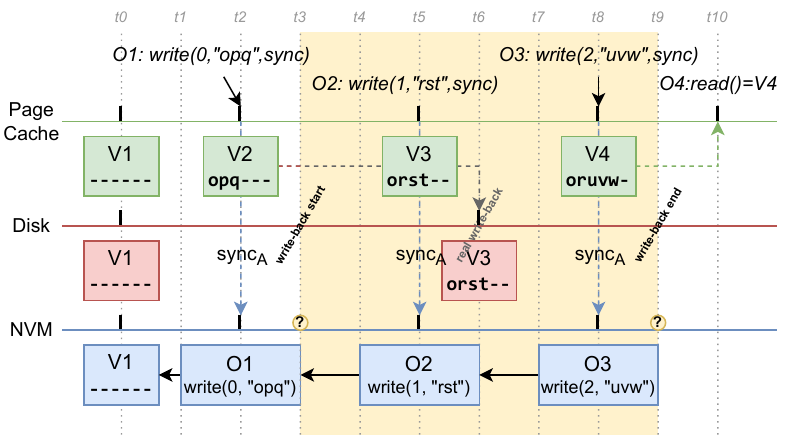}}
\caption{Space-time diagram for problematic cases caused by an uncertain write-back time. }
\label{fig3}
\end{figure*}

Fig. \ref{fig3} shows an example workflow in which the uncertain write-back time point may lead to the loss of write operations. At $t1$, both page cache, disk, and NVM reach a consistent state. 
Suppose that we only have arbitrary-length sync writes, because it's easy to see that ensuring the consistency of full-page sync is a relaxed subset of ensuring the consistency of arbitrary-length sync. 
At $t2$, $t5$, and $t8$ there are three sync writes generating $V2$, $V3$, and $V4$. At $t3$ the page encounters a write-back, but it's just the start of the write-back process. The data is written to the disk at $t6$, and $V3$ is the version written, because $O2$ happens after $wb^s$ ($t3$) before $wb^r$ ($t6$) and generate $V3$. Then at $t9$, the host is notified that the write-back is finished. Note that we can only mark the write-back event on the NVM (Algorithm \ref{alg2}) at the time of $wb^s$ or $wb^e$. Now when we inject crashes between the write-back process, some problems that we failed to discuss in Section \ref{sec:seq} and Section \ref{sec:grn} may occur.

\begin{scenario}\label{scene31}
Suppose that we crash before the real write-back, say, at $t4$. If we mark write-back on the NVM at $wb^s$ ($t3$), then until the crash time the data version that we expect to appear on the disk (no earlier than $V2$) has not been written to the disk yet, and the disk version is still $V1$. Then after the crash, if we follow Algorithm \ref{alg2} we will get \texttt{-rst--}, which is rebuilt from $V1\rightarrow O2$. Obviously, $O1$ is lost during the rebuild. So we can assume again that the write-back is marked at $wb^e$ ($t9$). Then if we crash at $t4$, the rebuild process can still generate $V2$ from $V1\rightarrow O1$, which is currently respecting Theorem \ref{theo:1}.
\end{scenario}

\begin{scenario}\label{scene32}
Suppose that we crash at $t10$ after the write-back ends, and the write-back will be marked on the NVM at $t9$, according to Scenario \ref{scene31}. Now when we rebuild the data after the crash, we will get $V3$ from the disk and nothing from the NVM, because everything on the NVM is before the write-back mark at $t9$ and is abandoned. Now it is obvious that $O3$ is lost, due to the postponement of the write-back mark. If we move the mark back to $wb^s$ ($t3$), it will work for this scenario, but will break Scenario \ref{scene31} as we just discussed. 
\end{scenario}

\begin{problem}\label{prob3}
The problems in Scenario \ref{scene31} and Scenario \ref{scene32} are caused by the uncertain write-back time between $wb^s$ and $wb^e$. Specifically, the reason is that we don't know which data version is the exact one being written to the disk, and when is it actually written. So when we choose $wb^s$ as the mark point, we may lose the data between the last write-back and $wb^s$ if the system crashes after $wb^s$ before $wb^r$. Otherwise when we choose $wb^e$ as the mark point, we may lose the data between $wb^r$ and $wb^e$ if the system crashes after $wb^e$. 
\end{problem}

\begin{solution}\label{sol3}
Since we know the start and end of a write-back, we can be sure that at least the versions before $wb^s$ can be abandoned once the write-back is successful, and the success of the write-back happens no later than $wb^e$. Thus we have the conclusion that at $wb^e$, the data before $wb^s$ is safely expired. Now when a write-back happens, we can temporarily record the time point (e.g. an auto-increment NVM record id) of $wb^s$ to prepare the data version that is going to expire. Then not until $wb^e$ will we make a real persistent record of the write-back event with the previously recorded time, indicating that eventually at this time point the data we recorded is confirmed to be useless. The process is shown in Algorithm \ref{alg3}, which is a modification to Algorithm \ref{alg2}. Maybe there are writes after $wb^s$ that are persisted by the disk at $wb^r$ (e.g. $V3$ in Fig. \ref{fig3}), there is no need to worry because NVPC ensures that data will be persisted by NVM when the page has a state of \texttt{DIRTY} or \texttt{WRITEBACK}, meaning that all sync writes between $wb^s$ and $wb^e$ will be recorded by NVM and replayed after crash. Thus persisting any data version between $wb^s$ and $wb^e$ to disk is acceptable, and any $(write,sync)$ pair crossing or inside the write-back is recoverable. 
\end{solution}

\begin{algorithm}
\caption{Solution to problem 3: record write-back events on the NVM with a proper design for $wb^s$ and $wb^e$. }
\label{alg3}

\begin{algorithmic}[1]

\Variables
    \State $prep\_rec[1\ldots page\_num]$ \Comment{Prepared write-back record for the later mark.}
    \State $page\_ver\_id[1\ldots page\_num]$ \Comment{An auto-increment id to track page version.}
\EndVariables

\Procedure{writeback}{$page$}   \Comment{This is $wb^s$.}

    \State $record.type\gets \mathrm{WRITEBACK}$ 
    \State $record.page\_index\gets page.index$
    \State $record.exp\_vid\gets page\_ver\_id[page.index]$
    \State $prep\_rec[page.index]\gets record$
    \State \Call{queue\_write\_disk}{$page$}

\EndProcedure

\Procedure{writeback\_callback}{$page$} \Comment{This is $wb^e$.}
    \State $prep\_rec[page.index].vid\gets page\_ver\_id[page.index]$
    \State $page\_ver\_id[page.index] = page\_ver\_id[page.index]+1$
    \State \Call{write\_nvm}{$prep\_rec[page.index]$}
\EndProcedure

\Procedure{sync\_write}{$page\_index$, $data$, $off$, $len$}
    
    \State $record.type\gets \mathrm{WRITE}$ 
    \State $record.page\_index\gets page\_index$
    \State $record.data\gets data$
    \State $record.off\gets off$
    \State $record.len\gets len$
    \State $record.vid\gets page\_ver\_id[page.index]$
    \State $page\_ver\_id[page.index] = page\_ver\_id[page.index]+1$
    \State \Call{write\_nvm}{$record$}

\EndProcedure

\Function{crash\_recover}{$page\_index$}
    \State $record \gets $ \Call{last\_nvm\_record\_on}{$page\_index$}
    \State $exp\_vid \gets 0$
    \State $useful\_records \gets [\;]$
    \State $i \gets 0$
    \While{$record \ne \mathrm{NULL}$ \textbf{and} $record.vid \ge exp\_vid$}
        \If{$record.type = \mathrm{WRITEBACK}$}
            \State $exp\_vid \gets record.exp\_vid$
        \Else
            \State $useful\_records[i] \gets record$
        \EndIf
        \State $i = i+1$
        \State $record \gets $ \Call{prev\_nvm\_record\_on\_same\_idx}{\allowbreak$record$}
    \EndWhile
    \State $page \gets $ \Call{read\_disk}{$page\_index$}
    \For{$i > 0$}
        \State $i = i-1$
        \State $record\gets useful\_records[i]$
        \State \Call{write\_page}{$page$, $record.data$, $record.off$, $record.len$}
    \EndFor
    \State \Return $page$
\EndFunction

\end{algorithmic}
\end{algorithm}

\begin{proposition}\label{prop3}
Solution \ref{sol3} solves Problem \ref{prob3}. 
\end{proposition}

\begin{proof}
For Problem \ref{prob3} and Solution \ref{sol3}, we can generate a sequence with atomic sync writes before $wb^s$, between $wb^s$ and $wb^r$, between $wb^r$ and $wb^e$, and after $wb^e$: $[w_1,sync,wb^s,w_2,sync,wb^r,w_3,sync,wb^e,w_4,sync]$. 
We inject $[crash,r]$ among each pair of operations to check if Solution \ref{sol3} works well. We also need to deal with the special case that a $wb^e$ happens between a pair of $[w,sync]$: $[w,wb^e,sync,crash,r]$. Then in the same way it is easy to prove that $wb^s$ or $wb^r$ inserted between $[w,sync]$ will not cause any fault. All the cases to be proved are listed below:

\begin{case}\label{case31}
Crashing before $wb^s$ is the same as the cases we discussed in the proof of Proposition \ref{prop1} and Proposition \ref{prop2}.
\end{case}

\begin{case}\label{case32}
Crashing between $[wb^s,w_2]$ generates $[w_1,sync,wb^s,crash,r]$. According to Algorithm \ref{alg3}, though we have prepared to mark $w_1$ as expired, since $wb^e$ never happens, the write-back record doesn't exist on the NVM. So when we rebuild the data, we use the previous disk data and the sync write records on the NVM, like $wb^s$ never happens. In this case, $(w_1,sync,crash,r)$ obviously respects Theorem \ref{theo:1}. In the same way we can prove that crashing between $[w_2,wb^r]$ respects Theorem \ref{theo:1}.
\end{case}

\begin{case}\label{case33}
Crashing between $[wb^r,w_3]$ generates $[w_1,sync,wb^s,w_2,sync,wb^r,crash,r]$. According to Algorithm \ref{alg3}, the NVM write-back record fails to be persisted, but the data version before $wb^r$, including $w_1$ and $w_2$, is written to the disk successfully. When we rebuild the data, we will replay all previous writes since the last write-back onto the disk data version that already contains them. This does not violate the sync semantics, because we are actually replaying all operations again, which contains each sync write and respects Theorem \ref{theo:1}. 
In the same way we can prove that crashing between $[w_3,wb^e]$ respects Theorem \ref{theo:1}.
\end{case}

\begin{case}\label{case34}
Crashing after $w_4$ generates 
$[w_1,\allowbreak sync,\allowbreak wb^s,\allowbreak w_2,\allowbreak sync,\allowbreak wb^r,\allowbreak w_3,\allowbreak sync,\allowbreak wb^e,\allowbreak w_4,\allowbreak sync,\allowbreak crash,\allowbreak r]$. 
According to Algorithm \ref{alg3}, after the recovery NVPC thinks that it has the data version before $wb^s$ on the disk, which includes $w_1$. But actually the data version on the disk includes both $w_1$ and $w_2$. The write operations $w_2$, $w_3$, and $w_4$ are persisted by the NVM, and will be replayed in sequence to the version on the disk. Though $w_2$ is replayed unnecessarily, it still respects Theorem \ref{theo:1}. Since all writes are persisted and replayed, each $(w,sync,crash,r)$ respects Theorem \ref{theo:1}. In the same way we can prove that crashing between $[wb^e,w_4]$ respects Theorem \ref{theo:1}.
\end{case}

\begin{case}\label{case35}
$[wb^s,w,wb^e,sync,crash,r]$. In this case, the write between $wb^s$ and $wb^e$ will re-dirty the page, so the sync can still persist it to NVM, and after the crash we can rebuild it because it happens after $wb^s$. Thus $(w,sync,crash,r)$ in this case respects Theorem \ref{theo:1}.
\end{case}

According to Theorem \ref{theo:2}, now that all cases respect Theorem \ref{theo:1}, Algorithm \ref{alg3} is proved to solve Problem \ref{prob3}.
\end{proof}

\section{Conclusion}

Heterogeneous storage devices have different characteristics that make system designers hard to rein. In this work, we take NVPC as the research target and discuss the difficulties of its data consistency, including the problems from the heterogeneity of write sequence, granularity, and mode, along with the solutions to the proposed problems.



\bibliographystyle{plain}

\vspace{12pt}

\end{document}